\documentclass{eptcs}
\providecommand{\event}{Mieke Massink and Gethin Norman (Eds.):
Ninth Workshop on Quantitative Aspects of Programming Languages (QAPL 2011)} 
 
\providecommand{\firstpage}{1} 
\usepackage{breakurl}        

\usepackage{graphicx}
\usepackage{color}
\usepackage{epsfig}
\usepackage{amssymb}
\usepackage{amsmath}
\usepackage{amsthm}
\usepackage[ansinew]{inputenc}
\usepackage{dsfont}
\usepackage{fancybox,tikz}

\newtheorem{proposition}{Proposition}
\newtheorem{definition}{Definition}

\newtheorem{example}{Example}

\title{Time Delays in Membrane Systems and Petri Nets}
\author{Bogdan Aman
\institute{A.I.Cuza University\\
Blvd. Carol I, no.11, 700506 Ia\c si, Romania}
\email{bogdan.aman@gmail.com}
\and Gabriel Ciobanu
\institute{Institute of Computer Science, Romanian Academy\\
and ''A.I.~Cuza'' University of Ia\c si, Romania}
\email{gabriel@info.uaic.ro} }
\def\titlerunning{Time Delays in Membrane Systems and Petri Nets}
\def\authorrunning{B. Aman and G. Ciobanu}
\begin{document}
\maketitle

\begin{abstract}

Timing aspects in formalisms with explicit resources and parallelism are
investigated, and it is presented a formal link between timed membrane systems
and timed Petri nets with localities. For both formalisms, timing does not
increase the expressive power; however both timed membrane systems and timed
Petri nets are more flexible in describing molecular phenomena where time is a
critical resource. We establish a link between timed membrane systems and timed
Petri nets with localities, and prove an operational correspondence between
them.

\end{abstract}

\section{Introduction}
\label{section:introduction}

The evolution of complex real systems frequently involves various
interactions among components. Some mathematical models of such systems
combine both discrete and continuous evolutions on multiple time scales
with many orders of magnitude. For example, the molecular operations of a
living cell can be thought of as such a dynamical system. The molecular
operations happen on time scales ranging from $10^{-15}$ to $10^4$ seconds,
and proceed in ways which are dependent on populations of molecules ranging
in size from as few as approximately $10$ to approximately as many as
$10^{20}$. Molecular biologists have used formalisms developed in computer
science (e.g. hybrid Petri nets) to get simplified models of some molecular
phenomena like transcription and gene regulation processes. According to
molecular cell biology \cite{Lodish08}:
(i) ``the life span of intracellular proteins varies from as
short as a few minutes for mitotic cycles, which help regulate
passage through mitosis, to as long as the age of an organism for
proteins in the lens of the eye'', and
(ii) ``Most cells in multicellular organisms $\dots$ carry
out a specific set of functions over periods of days to months or
even the lifetime of the organism (nerve cells, for example)''.
Lifetimes play an important role in the
biological evolution; we mention an example from the immune system.
\begin{example}
According to \cite{Lodish08}, T-cell precursors arriving in the thymus from
the bone marrow spend up to a week differentiating there before they enter
a phase of intense proliferation. In a young adult mouse the thymus
contains around $10^8$ to $2 \times 10^8$ thymocytes. About $5 \times 10^7$
new cells are generated each day; however, only about $10^6$ to $2 \times
10^6$ (roughly $2-4\%$) of these will leave the thymus each day as mature T
cells. Despite the disparity between the numbers of T cells generated daily
in the thymus and the number leaving, the thymus does not continue to grow
in size or cell number. This is because approximately $98\%$ of the
thymocytes which develop in the thymus also die within the thymus.
\end{example}

Among the formalisms able to model these systems by using explicit
resources, parallelism and timing, we refer to membrane systems
\cite{Paun02} and Petri nets \cite{Jensen92,Peterson81}. Membrane
systems were extended with timing aspects in
\cite{Cavaliere05,Cavaliere10}. Petri Nets have two main extensions
with time: Time Petri Nets \cite{Merlin74} (a transition can fire
within a time interval) and Timed Petri Nets \cite{Ramchandani74} (a
transition fires as soon as possible). In Petri nets, time
can be considered relative both to places and transitions
\cite{Pezze99,Sifakis80}. In this paper, we define a timed extension
(relative to transitions) for Petri nets with localities, and we
establish a link between timed membrane systems and timed Petri nets
with localities.

Some connections between membrane systems and Petri nets are presented for the
first time in \cite{Zilio04,Qi04}. A direct structural relationship between
these two formalisms is established in \cite{Kleijn10,Kleijn06} by defining a
new class of Petri nets called Petri nets with localities. Localities are used
to model the regions of membrane systems. This new class of Petri
nets has been used to show how maximal evolutions from membrane systems are
faithfully reflected in the maximally concurrent step sequence semantics of
their corresponding Petri nets with localities.

Despite the fact that various timed extensions exist for both
membrane systems and Petri nets, we are not aware of any connection
between these timed extensions. Thus, we relate timed membrane
systems with timed Petri nets with localities. The existing links
(marked by citation or easy to prove) between timed membrane systems
and timed Petri nets are described in the following diagram.

\medskip

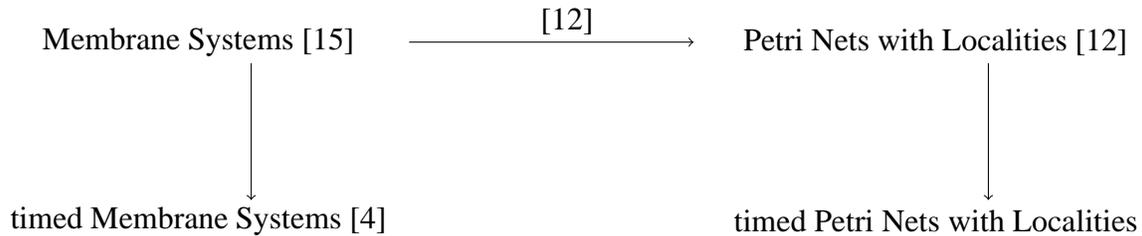
\begin{center}
\begin{tikzpicture}[scale=1.4]
\node at (0.0,0.0) {\large Membrane Systems \cite{Paun02}};

\node at (7.0,0.0) {\large Petri Nets with Localities \cite{Kleijn06}};

\draw[->] (2.0,0.0) -- (4.7,0.0);

\node at (3.5,0.2) {\large \cite{Kleijn06}};

\draw[->] (0.5,-0.2) -- (0.5,-1.5);

\draw[->] (7.5,-0.2) -- (7.5,-1.5);

\node at (0.0,-1.7) {\large timed Membrane Systems
\cite{Cavaliere05}};

\node at (7.0,-1.7) {\large timed Petri Nets with Localities};

\end{tikzpicture}
\end{center}

\medskip

\noindent Surprisingly, we prove that adding timing aspects does
not lead to more powerful formalisms, and the new links are expressed
by the following diagram.

\medskip

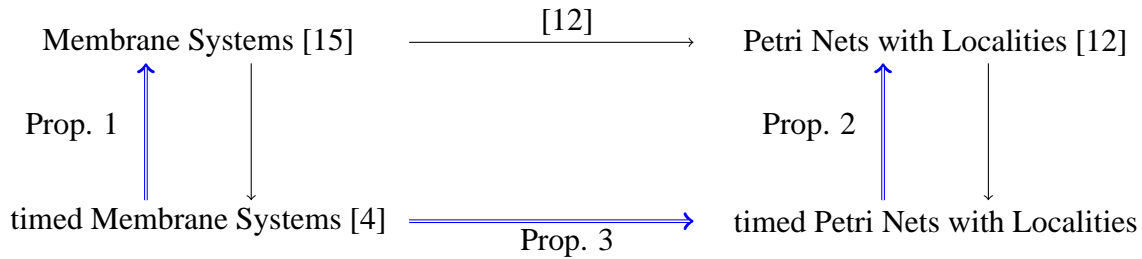
\begin{center}
\begin{tikzpicture}[scale=1.4]
\node at (0.0,0.0) {\large Membrane Systems \cite{Paun02}};

\node at (7.0,0.0) {\large Petri Nets with Localities \cite{Kleijn06}};

\draw[->] (2.0,0.0) -- (4.7,0.0);

\node at (3.5,0.2) {\large \cite{Kleijn06}};

\draw[->] (0.5,-0.2) -- (0.5,-1.5);

\draw[<-,double,blue] (-0.5,-0.2) -- (-0.5,-1.5);

\node at (-1.2,-0.8) {\large Prop. \ref{tPtoP}};

\draw[->] (7.5,-0.2) -- (7.5,-1.5);

\draw[<-,double,blue] (6.5,-0.2) -- (6.5,-1.5);

\node at (5.8,-0.8) {\large Prop. \ref{tPNtoPN}};

\node at (0.0,-1.7) {\large timed Membrane Systems
\cite{Cavaliere05}};

\node at (7.0,-1.7) {\large timed Petri Nets with Localities};

\draw[->,double,blue] (2.0,-1.7) -- (4.7,-1.7);

\node at (3.5,-1.9) {\large Prop. \ref{proposition:corresp}};

\end{tikzpicture}
\end{center}

\medskip

We prove that timing does not increase the expressive power of both membrane
systems and Petri nets with localities. However the timed formalisms are able
to describe more naturally some real systems involving timing. Although there
are few extensions with time for both membrane systems and Petri nets, it does
not exist a connection between these timed extensions. An attempt is presented
in \cite{Profir05} by using a software simulation (and having some decidability
aims). We relate timed membrane systems to timed Petri nets with localities
following the research line of \cite{Kleijn06}, and prove an operational
correspondence between them.

\section{Timed Membrane Systems}
\label{section:timed_symanti}

Membrane systems (also called P systems) are introduced by P\u aun as a
model of distributed, parallel and nondeterministic systems inspired by
cell biology \cite{Paun02}. A cell is divided in various compartments, each
compartment with a different task, with all of them working simultaneously
to accomplish a more general task for the whole system. The membranes
determine regions where objects and evolution rules can be placed. The
objects evolve according to the rules associated with each region, and the
regions cooperate in order to maintain the proper behaviour of the whole
system. The application of evolution rules is done in parallel, and is
eventually regulated by priority relationships between rules. Several
results and variants of membrane systems (inspired by different aspects of
living cells like symport and antiport communication through membranes,
catalytic objects, membrane charge, etc.) are presented in~\cite{Paun02}.
Various applications of membrane systems are presented in \cite{Ciobanu06}.
Links between membrane systems and process calculi are presented in
\cite{Ciobanu10}. An updated bibliography can be found on the membrane
systems webpage {\sf http://ppage.psystems.eu}.

The structure of a membrane system is represented by a tree (with the skin as
its root), or equivalently, by a string of correctly matching parentheses where
each pair of matching parentheses corresponds to a membrane. Graphically, a
membrane structure is represented by a Venn diagram in which two sets can be
either disjoint, or one is the subset of the other. A membrane without any other
membrane inside is said to be elementary. The membranes are labelled in a
one-to-one manner.

Let $\mathds{N}$ be the set of positive integers, and $V$ a finite alphabet of
symbols. A multiset over $V$ is a mapping $u: V \to \mathds{N}$. We use the
string representation of multisets that is widely accepted and used in membrane
systems; a multiset $w$ described by $a^2b^5$ means that $a$ appears twice in
$w$, while $b$ appears five times in $w$. We use a global clock to simulate the
passage of time. The following definition of timed membrane systems is similar
to that introduced in \cite{Cavaliere05}, but without considering catalysts,
signal-promoters and output region.

\begin{definition}
\label{definition:timed_symanti} A {\rm timed membrane system}
$\Pi=(V,\mu,w_1,\ldots,w_n,R_1,\ldots,R_n,e)$ is defined by
\begin{itemize}
\item[$\bullet$] $V$ is an alphabet (its elements are called
{\rm objects});

\item[$\bullet$] $\mu$ describes the {\rm membrane
structure}, namely a structure consisting of a hierarchy of $n$
membranes labelled from $1$ to $n$ which are either disjoint or
included; we distinguish the external membrane, usually called
``skin'';

\item[$\bullet$] $w_1,\ldots,w_n$ are finite multisets over $V$; $w_i$
represents the multiset of objects associated to membrane $i$; $n
\geq 1$ is the initial degree of the system;

\item[$\bullet$] $R_1,\ldots,R_n$ are finite sets of
evolution rules over $V$ associated with the membranes of $\mu$; the
rules are of the form $a \rightarrow v$, where $a \in V$ and $v$ is
a multiset from $\{(a,here),(a,out)\mid a\in V\} \cup \{(a,in_j)
\mid a\in V, 1 \leq j \leq n\}$;

\item $e: R_1 \cup \ldots \cup R_n \rightarrow \mathds{N}$ is a
(computable) function indicating the execution time of each
evolution rule; the time evolves according to a global clock that
starts from 0 and splits time in equal intervals (units of time).
\end{itemize}
\end{definition}

The membrane structure and the multisets in $\Pi$ determine a
configuration of the system. We can pass from a configuration to
another one by using the evolution rules. The use of a rule
$u\rightarrow v$ in a region with a multiset $w$ means to subtract
the multiset identified by $u$ from $w$, and then add the multiset
represented by~$v$.
Since the right hand side $v$ of a rule consists only of messages,
an object introduced by a rule cannot evolve in the same step by
means of another rule. If a message appears in $v$ in the form
$(c,here)$, then it remains in the same region. If it appears as
$(c,in_j)$, then a copy of $c$ is introduced in the child membrane
with the label $j$; if a child membrane with the label $j$ does not
exist, then the rule cannot be applied. If it appears as $(c,out)$,
then a copy of the object $c$ is introduced in the parent (surrounding)
membrane. The system may contain rules which are never applicable,
and also rules which send objects out of the skin.

The evolution rules in a membrane are applied in a
maximal parallel manner, and all membranes evolves in parallel. At
each tick of the (global) clock, all the rules that can be applied
must be applied in a maximal parallel manner (this means that no
further rule could be applied at the same time unit). An evolution
rule $r$ started at the $j$-th tick of the clock ends its execution
at the $j+e(r)$-th tick, meaning that the newly created objects by
rule $r$ can be used starting from the $j+e(r)+1$-th tick of the
clock. When a rule starts, the objects from the left hand side of
the rule become unavailable for other rules.

\medskip

\begin{figure}[ht]
\begin{tabular}{c@{\hspace{6ex}}c}
\begin{tikzpicture}[scale=1.2]
\node at (1.0,3.0) {$R_2=\{r_3:a \rightarrow (b,out)(a,here)\}$};

\node at (1.0,2.5) {$\quad\cup\;\{r_4:b \rightarrow (b,out)\}$};

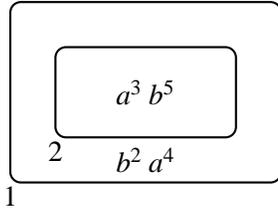
\node at (1.0,4.0) {$R_1= \{r_1:a \rightarrow (a,in_2)\}$};

\node at (1.0,3.5) {$\quad\cup\;\{r_2:b \rightarrow (a,in_2)\}$};

\draw[thick,rounded corners=4pt] (0.0,0.5) rectangle (2.0,1.5);

\node at (0.0,0.35) {$2$};

\node at (1.0,1.0) {$a^3\;b^5$};

\draw[thick,rounded corners=4pt] (-0.5,0.0) rectangle (2.5,2.0);

\node at (-0.5,-0.15) {$1$};

\node at (1.0,0.25) {$b^2\;a^4$};
\end{tikzpicture}&
\begin{minipage}{7.7cm}\vspace{-33ex}{As an example, we consider
a membrane system with two nested membranes (the inner membrane
labelled by $2$, the outer membrane labelled by $1$), two sets
$R_1$ and $R_2$ of evolution rules having the execution times $e(r_1)=2$,
$e(r_2)=5$, $e(r_3)=3$, $e(r_4)=1$, a global clock and two symbols
($a$ and $b$). Initially, membrane~$1$ contains the multiset
$b^2\;a^4$, and membrane~$2$ contains the multiset $a^3\;b^5$.}

\end{minipage}\\
\end{tabular}
\centering\vspace{-3ex}\caption{A Timed Membrane System}
\label{figure:example_membrane}
\end{figure}

In what follows we define the configurations of a membrane system,
and the transition system given by considering each of the
transition steps defined by maximally parallel rewriting and
parallel communication, as in \cite{CiobanuHandbook}. Let $V$ be a
finite alphabet of objects over which we consider the free
commutative monoid $V^*$ whose elements are multisets (the empty
multiset is denoted by $\varepsilon$). Objects together with a target indication
are enclosed in messages of form $(w, here)$, $(w, out)$, and
$(w, in_l )$. For the sake of simplicity, hereinafter we consider
that the messages with the same target indication merge into one
message:

\centerline{$\displaystyle \prod_{i\in I} (v_i,here)=(w,here)$,
$\displaystyle \prod_{i\in I} (v_i,in_l)=(w,in_l)$, $\displaystyle
\prod_{i\in I} (v_i,out)=(w,out)$,}

\noindent
with $w=\displaystyle \prod_{i\in I} v_i$, $I$ a non-empty
set, and $(v_i)_{i\in I}$ a family of multisets over $V$.

A configuration for a membrane system is a tuple
$C=(w_1,\ldots,w_n,k)$, namely the multisets of all regions together
with the value of the global clock. An intermediary configuration is
a tuple in which the objects have associated target indications.
Each membrane system has an initial configuration which is
characterized by the initial multiset of objects for each membrane
of the initial membrane structure of the system.
For two configurations $C$ and $C'$ of $\Pi$, we say that there is a
transition from $C$ to $C'$, and write $C\Rightarrow C'$, if the
following {\it steps} are executed in the given order:
\begin{enumerate}
\item {\em maximal parallel rewriting step}
($\stackrel{mpr}{\Longrightarrow}$):
each membrane evolves in a maximal parallel manner;
\item {\em parallel communication of objects through membranes}
($\stackrel{tar}{\Longrightarrow}$), by sending and receiving messages.
\end{enumerate}
The last step takes place only if there are messages resulting from
the first step. If the first step is not possible, then neither is
the second step, and we say that the system has reached a {\em halting
configuration}. According to \cite{Andrei07}, a transition step
between two configurations $C,C'$ is given by: $C \Rightarrow C'$
iff $C$ and $C'$ are related by the following relation: $C
\stackrel{mpr}{\Longrightarrow} \stackrel{tar}{\Longrightarrow}C'$.
Starting from a configuration without messages, we apply the ``mpr''
step and get an intermediate configuration; if we have messages,
then we apply the ``tar'' step. If the last configuration has no
messages, then we say that the transition relation $\Rightarrow$ is
well-defined as an evolution step between the first and last
configurations.

The evolution of the system $\Pi$ at time step $k$, from a
configuration $C=(w_1,\ldots,w_n,k)$ to another configuration
$C'=(w'_1,\ldots,w'_n,k+1)$ is made by applying
a multiset of rules $R$ in a maximally parallel manner.
If the multiset $R$ of rules is empty, then
only the clock is incremented (from $k$ to $k+1$).
Given a multiset of rules $R$,
we denote by $lhs_i=\sum_{r\in R} R(r) \cdot lhs^r_i$ the multiset
of objects in the left hand sides of the rules in $R$ which are
associated to membrane $i$. In a similar way, by
$rhs_{i,j}^k = \sum_{r\in R;\;e(r)=j} R(r)\cdot~rhs^{r,k}_{i,j}$ is
denoted the multiset of objects in the right hand sides of the rules
in $R$ applied at time $k$ which is associated to membrane $i$ after
$j$ units of time. We also denote by $m=max_{r\in R}e(r)$ the
maximum delay inferred by the rules of~$R$. $C$ evolves to $C'$ by a
multiset $R$ of rules (this is denoted by
$C\stackrel{R}{\Longrightarrow}C'$) if for each membrane $i$ the
following conditions~hold:
\begin{enumerate}
\item[$(i)$] $lhs_i \leq w_i$;
\item[$(ii)$] there is no rule $r\not \in R$ such that
$lhs_i^r+ lhs_i\leq w_i$;
\item[$(iii)$] for each $a\in V$, $w'_i(a)=w_i(a)-lhs_i(a)+
\sum^{k}_{s=max(0,k-m)} rhs^s_{i,0}(a)$.
\end{enumerate}
According to $(i)$, a configuration $C$ has in each membrane labelled by $i$
enough objects to enable the execution of the multiset $R$ of rules.
The maximal parallelism is captured by $(ii)$, saying that an extra
evolution rule cannot be added to $R$. Condition $(iii)$ describes
the effect of the rules application by adding all the objects having
$j=0$ created in the last $min(k,m)$ steps which are ready to be
used in the membrane system evolution. Before incrementing the
global clock, all multisets $rhs^s_{i,j}$ are transformed into
$rhs^{s}_{i,j-1}$ for $max(0,k-m) \leq s,j \leq k$.

\begin{proposition}\label{tPtoP}
For every timed membrane system $\Pi=(V,\mu,w_{0,1},\ldots,w_{0,n},
R_1, \ldots, R_n, e)$ there exists an untimed membrane system
$\Pi'=(V',\mu',w'_{0,1},\ldots,w'_{0,n},R'_1,\ldots,R'_n)$ that
simulates the evolution of $\Pi$ (restricted to the elements of
$V$). Formally, for all $a\in V$ and $k \in \mathds{N}$ we have
$w_{k,i}(a)=w'_{k,i}(a)$, where $w_{k,i}$ and $w'_{k,i}$ are the multisets of
objects from membrane $i$ of $\Pi$ and $\Pi'$ at step~$k$.
\end{proposition}

\begin{proof} In what follows we show how starting from
a timed membrane system $\Pi=(V,\mu,w_{0,1},\ldots,w_{0,n},R_1,$
$\ldots,R_n,e)$ we may construct an untimed membrane system
$\Pi'=(V',\mu',w'_{0,1},\ldots,w'_{0,n},R'_1,\ldots,R'_n)$, where
\begin{itemize}
\item $V'=V \cup \{a_j\mid a\in V, 0 \leq j \leq m-1\}$, where $m=max_{r\in
R}e(r)$;
\item $\mu'=\mu$ and $w'_{0,i}=w_{0,i}$ for $1 \leq i \leq n$;
\item for each rule $r:u \rightarrow v$ of $R_i$, $1 \leq i \leq
n$ having $e(r)= 0$, we add $r$ to $R'_i$;
\item for each rule $r:u \rightarrow v$ of $R_i$, $1 \leq i \leq
n$ having $e(r)> 0$, we add to $R'_i$ the following sets of rules
which simulate properly the passage of $e(r)$ units of time:
\begin{itemize}
\item $u \rightarrow v'$, where $v'$ is derived from $v$ by
replacing each $a\in V$ by $a_{e(r)-1}\in V'$;
\item $a_j \rightarrow a_{j-1}$, $1 \leq j \leq e(r)-1$;
\item $a_0 \rightarrow a$.
\end{itemize}
\end{itemize}
We show that each step of the timed membrane system can be simulated
by the corresponding untimed membrane system, using induction
on the number of steps (time units) in timed membrane system.

Firstly, we consider a configuration
$C_0=(w_{0,1},\ldots,w_{0,n},0)$ of the timed membrane system and a
maximal multiset $R$ of rules such that
$C_0\stackrel{R}{\Longrightarrow}C_{1}$. The resulting configuration
$C_{1}=(w_{1,1},\ldots,w_{1,n},1)$ is given by
$w_{1,i}(a)=w_{0,i}(a)-lhs(i)(a)+rhs^0_{i,0}(a)$ for all $1\leq
i\leq n$ and $a \in V$. Following the construction above, the
initial configuration of the untimed membrane system is
$C'_0=(w'_{0,1},\ldots,w'_{0,n})$ where $w'_{0,i}(a)=w_{0,i}(a)$ for
all $1\leq i\leq n$ and $a \in V$. $R'$ is the multiset of rules
obtained from $R$ such that
$C'_0\stackrel{R'}{\Longrightarrow}C'_{1}$. The resulting
configuration $C'_1$ is given by
$w'_{1,i}(a)=w'_{0,i}(a)-lhs_{i}(a)+rhs_i(a)$ for all $1\leq i\leq
n$ and $a \in V'$. This configuration contains all the elements of
$C_1$ and some additional objects from $V'$ introduced to simulate
properly the passage of time. Regarding the elements $a\in V$, it
results that $rhs^0_{i,0}(a)=rhs_i(a)$, namely
$w'_{1,i}(a)=w_{1,i}(a)$. Therefore $C'_1$ equals $C_1$ regarding
the elements of $V$ (we ignore the new elements of $V'$
because they are used only to simulate the passage of time).

Secondly, we consider a configuration
$C_k=(w_{k,1},\ldots,w_{k,n},k)$ of the timed membrane system and a
maximal multiset $R$ of rules such that
$C_k\stackrel{R}{\Longrightarrow}C_{k+1}$. The resulting
configuration $C_{k+1}=(w_{(k+1),1},\ldots,$ $w_{(k+1),n},k+1)$ is
given by
$w_{(k+1),i}(a)=w_{k,i}(a)-lhs(i)(a)+\sum^{k}_{s=max(0,k-m)}
rhs^s_{i,0}(a)$ for all $1\leq i\leq n$ and $a \in V$. In the same
time, the multisets $rhs^s_{i,j}$ are transformed into
$rhs^{s}_{i,j-1}$ for $max(0,k-m) \leq s,j \leq k$. Following the
construction above, the configuration of the untimed membrane system
is $C'_k=(w'_{k,1},\ldots,w'_{k,n})$, where $w'_{k,i}(a)=w_{k,i}(a)$
for all $a\in V$, and $w'_{k,i}(a_j)=\sum^{k}_{s=max(0,k-m)}
rhs^{s}_{i,j}(a)$ for all $a_j \in V'\backslash V$. This means that
for all $1 \leq i \leq n$, the multiset $w'_{k,i}$ contains all the
objects from $w_{k,i}$ and some additional objects from $V'$. For
each $a\in V$ from the multiset $rhs^s_{i,j}$, the multiset $w'_{k,i}$
contains additional objects $a_j$. The restriction $max(0,k-m) \leq
s\leq k$ used when creating the object $a_j$ in membrane $i$ means
that an object $a$ has appeared in the right hand side of a rule
from timed membrane systems in the last $min(k,m)$ units of time,
but has to wait $j$ units of time until it should be added to
membrane $i$ in timed membrane systems. $R'$ is the multiset of
rules obtained from $R$ such that
$C'_k\stackrel{R'}{\Longrightarrow}C'_{k+1}$, with
$w'_{k+1,i}(a)=w'_{k,i}(a)-lhs_{i}(a)+rhs_i(a)$ for all $1\leq i\leq
n$ and $a \in V'$. Moreover, in this step some objects of $V'$ are
transformed into objects of $V$ by applying the generic rule $a_0
\rightarrow a$ (the other objects $a_j \in V'$ are transformed into
objects $a_{j-1} \in V'$ by applying the generic rules $a_j
\rightarrow a_{j-1}$). Finally, the number of objects $a\in V$
obtained in $\Pi'$ at this step corresponds to
$\sum^{k}_{s=max(0,k-m)}rhs^s_{i,0}(a)$. It results that
$\sum^{k}_{s=max(0,k-m)} rhs^s_{i,0}(a)=rhs_i(a)$, namely
$w'_{(k+1),i}(a)=w_{(k+1),i}(a)$. Therefore $C'_{k+1}$ equals
$C_{k+1}$ regarding the elements of $V$ (we ignore the
elements $a_j\in V'$ because they are used only to simulate the passage of time).
\end{proof}

In what follows we give an example that illustrates the statement of
Proposition \ref{tPtoP}.

\begin{example}
We consider a timed membrane system $\Pi=(V,\mu,w_1,w_2,R_1,R_2,e)$,
where:
\begin{itemize}
\item $V=\{a,b\}$; \qquad $\mu=[[~]_2]_1$; \qquad
$w_1=ab$; \qquad $w_2=a^2b$;
\item $R_1=\{r_1: b \rightarrow (b,in_2)\}$; \qquad
$R_2=\{r_2: a \rightarrow (a,out)\}$; \qquad
 $e(r_1)=0$, $e(r_2)=2$.
\end{itemize}

\noindent Since the initial configuration of the timed membrane
system $\Pi$ is $(ab,a^2b,0)$, then the evolution of the timed
membrane system in terms of configurations is:

$(ab,a^2b,0)$ $\stackrel{\{r_1+2r_2\}}{\quad\Longrightarrow\quad}$ $(a,b^2,1)$
$\stackrel{\emptyset}{\quad\Longrightarrow\quad}$ $(a,b^2,2)$
$\stackrel{\emptyset}{\quad\Longrightarrow\quad}$  $(a^3,b^2,3)$

\noindent Graphically this can be depicted as:

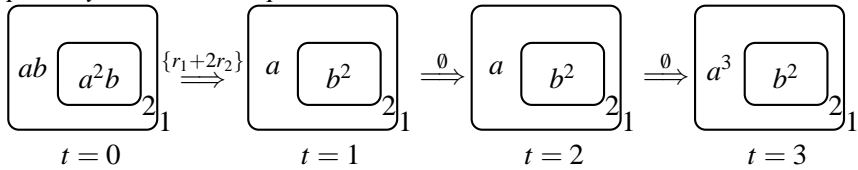
\begin{tikzpicture}[scale=1.1]
\begin{scope}[xshift=-0.2cm]
\draw[thick,rounded corners=4pt] (1.0,0.5) rectangle (2.8,2.0);

\node at (2.9,0.6) {$1$};

\node at (1.3,1.25) {$ab$};

\node at (3.35,1.25) {$\stackrel{\{r_1+2r_2\}}{\Longrightarrow}$};

\begin{scope}[xshift=1.6cm,yshift=-0.2cm]
\draw[thick,rounded corners=4pt] (0.0,1.0) rectangle (1.0,1.75);

\node at (1.1,1.0) {$2$};

\node at (0.5,1.35) {$a^2b$};
\end{scope}

\node at (2,0.2) {$t=0$};
\end{scope}

\begin{scope}[xshift=2.7cm]
\draw[thick,rounded corners=4pt] (1.0,0.5) rectangle (2.8,2.0);

\node at (2.9,0.6) {$1$};

\node at (1.3,1.25) {$a$};

\node at (3.35,1.25) {$\stackrel{\emptyset}{\Longrightarrow}$};

\begin{scope}[xshift=1.6cm,yshift=-0.2cm]
\draw[thick,rounded corners=4pt] (0.0,1.0) rectangle (1.0,1.75);

\node at (1.1,1.0) {$2$};

\node at (0.5,1.35) {$b^2$};
\end{scope}

\node at (2,0.2) {$t=1$};
\end{scope}

\begin{scope}[xshift=5.4cm]
\draw[thick,rounded corners=4pt] (1.0,0.5) rectangle (2.8,2.0);

\node at (2.9,0.6) {$1$};

\node at (1.3,1.25) {$a$};

\node at (3.35,1.25) {$\stackrel{\emptyset}{\Longrightarrow}$};

\begin{scope}[xshift=1.6cm,yshift=-0.2cm]
\draw[thick,rounded corners=4pt] (0.0,1.0) rectangle (1.0,1.75);

\node at (1.1,1.0) {$2$};

\node at (0.5,1.35) {$b^2$};
\end{scope}

\node at (2,0.2) {$t=2$};
\end{scope}

\begin{scope}[xshift=8.1cm]
\draw[thick,rounded corners=4pt] (1.0,0.5) rectangle (2.8,2.0);

\node at (2.9,0.6) {$1$};

\node at (1.3,1.25) {$a^3$};

\begin{scope}[xshift=1.6cm,yshift=-0.2cm]
\draw[thick,rounded corners=4pt] (0.0,1.0) rectangle (1.0,1.75);

\node at (1.1,1.0) {$2$};

\node at (0.5,1.35) {$b^2$};
\end{scope}

\node at (2,0.2) {$t=3$};
\end{scope}
\end{tikzpicture}

\noindent We construct an untimed membrane system
$\Pi'=(V',\mu',w'_1,w'_2,R'_1,R'_2)$, where:
\begin{itemize}
\item $V'=\{a,a_0,a_1,b,b_0,b_1\}$; \qquad $\mu=[[~]_2]_1$; \qquad
$w_1=ab$; \qquad $w_2=a^2b$;
\item $R_1=\{r_1: b \rightarrow (b,in_2)\}$; \qquad
$R_2=\{r^1_2: a \rightarrow (a_1,out);\;r^2_2: a_1 \rightarrow
a_0;\;r^3_2: a_0 \rightarrow a\}$.
\end{itemize}

\noindent Since the initial configuration of the untimed membrane
system $\Pi'$ is the same as the initial configuration of the timed
membrane system $\Pi$, namely $(ab,a^2b,0)$, then the evolution of
the untimed membrane system in terms of configurations is:

$(ab,a^2b)$ $\stackrel{\{r_1+2r^1_2\}}{\quad\Longrightarrow\quad}$
$(aa^2_1,b^2)$ $\stackrel{\{r^2_2\}}{\quad\Longrightarrow\quad}$ $(aa^2_0,b^2)$
$\stackrel{\{r^3_2\}}{\quad\Longrightarrow\quad}$  $(a^3,b^2)$

\noindent Graphically this can be depicted as:

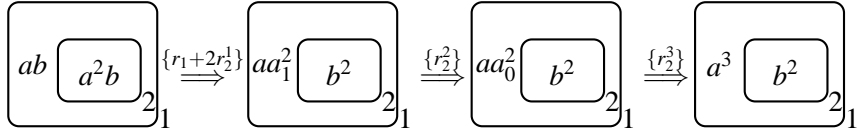
\begin{tikzpicture}[scale=1.1]
\begin{scope}[xshift=-0.2cm]
\draw[thick,rounded corners=4pt] (1.0,0.5) rectangle (2.8,2.0);

\node at (2.9,0.6) {$1$};

\node at (1.3,1.25) {$ab$};

\node at (3.35,1.25) {$\stackrel{\{r_1+2r^1_2\}}{\Longrightarrow}$};

\begin{scope}[xshift=1.6cm,yshift=-0.2cm]
\draw[thick,rounded corners=4pt] (0.0,1.0) rectangle (1.0,1.75);

\node at (1.1,1.0) {$2$};

\node at (0.5,1.35) {$a^2b$};
\end{scope}

\end{scope}

\begin{scope}[xshift=2.7cm]
\draw[thick,rounded corners=4pt] (1.0,0.5) rectangle (2.8,2.0);

\node at (2.9,0.6) {$1$};

\node at (1.3,1.25) {$aa_1^2$};

\node at (3.35,1.25) {$\stackrel{\{r^2_2\}}{\Longrightarrow}$};

\begin{scope}[xshift=1.6cm,yshift=-0.2cm]
\draw[thick,rounded corners=4pt] (0.0,1.0) rectangle (1.0,1.75);

\node at (1.1,1.0) {$2$};

\node at (0.5,1.35) {$b^2$};
\end{scope}

\end{scope}

\begin{scope}[xshift=5.4cm]
\draw[thick,rounded corners=4pt] (1.0,0.5) rectangle (2.8,2.0);

\node at (2.9,0.6) {$1$};

\node at (1.3,1.25) {$aa_0^2$};

\node at (3.35,1.25) {$\stackrel{\{r^3_2\}}{\Longrightarrow}$};

\begin{scope}[xshift=1.6cm,yshift=-0.2cm]
\draw[thick,rounded corners=4pt] (0.0,1.0) rectangle (1.0,1.75);

\node at (1.1,1.0) {$2$};

\node at (0.5,1.35) {$b^2$};
\end{scope}

\end{scope}

\begin{scope}[xshift=8.1cm]
\draw[thick,rounded corners=4pt] (1.0,0.5) rectangle (2.8,2.0);

\node at (2.9,0.6) {$1$};

\node at (1.3,1.25) {$a^3$};

\begin{scope}[xshift=1.6cm,yshift=-0.2cm]
\draw[thick,rounded corners=4pt] (0.0,1.0) rectangle (1.0,1.75);

\node at (1.1,1.0) {$2$};

\node at (0.5,1.35) {$b^2$};
\end{scope}

\end{scope}
\end{tikzpicture}

\noindent If we are interested only in the symbols of $V$
in the untimed evolution, then we have:

\begin{tabular}{c@{\hspace{0ex}}c@{\hspace{0ex}}c@{\hspace{0ex}}c
@{\hspace{0ex}}c@{\hspace{0ex}}c@{\hspace{0ex}}c@{\hspace{0ex}}c
@{\hspace{0ex}}c@{\hspace{0ex}}c @{\hspace{0ex}}c
@{\hspace{0ex}}c@{\hspace{0ex}}c@{\hspace{0ex}}c@{\hspace{0ex}}c}

$(ab$,&$a^2b$,&$0)$ &$\stackrel{\{r_1+2r_2\}}{\Longrightarrow}$&
$(a$,&$b^2$,&$1)$ &$\stackrel{\emptyset}{\Longrightarrow}$
&$(a$,&$b^2$,&$2)$&
$\stackrel{\emptyset}{\Longrightarrow}$ & $(a^3$,&$b^2$,&$3)$\\

$||$ &$||$&&& $||$ &$||$&&& $||$ &$||$ &&& $||$&$||$\\

$(ab \cap V$,& $a^2b\cap V)$&&
$\stackrel{\{r_1+2r^1_2\}}{\Longrightarrow}$& $(aa^2_1\cap
V$,&$b^2\cap V)$& &$\stackrel{\{r^2_2\}}{\Longrightarrow}$&
$(aa^2_0\cap V$,&$b^2\cap V)$&
&$\stackrel{\{r^3_2\}}{\Longrightarrow}$&
$(a^3\cap V$,&$b^2\cap V)$&\\
\end{tabular}

\noindent and thus the statement of Proposition \ref{tPtoP} holds.

\end{example}

It is easy to prove that the class of timed membrane systems
includes the class of untimed membrane systems, since we can assign
$0$ to all the rules by the timing function $e$.

\section{Timed Petri Nets with Localities}
\label{subsection:cpn}

An extension of Petri nets with localities is defined by adding
delays to transitions (like in coloured Petri nets \cite{Jensen92}).
The value of the global clock is kept in a variable $gc$.

\begin{definition} A timed Petri net with localities
$\mathcal{N}=(P,T,W,L,D,M_0)$ is given~by:
\begin{enumerate}
\item[$(i)$] finite disjoint sets $P$ of places and $T$ of
transitions;
\item[$(ii)$] a weight function $W: (T \times P) \cup (P \times T)
\rightarrow \mathds{N}$;
\item[$(iii)$] a locality mapping $L: T \rightarrow \mathds{N}$;
\item[$(iv)$] a delay mapping $D: T \rightarrow \mathds{N}$;
\item[$(v)$] an initial marking $M_0: P \cup \{gc\} \rightarrow \mathds{N}$.
\end{enumerate}
\end{definition}
\noindent If $W(x,y)\geq 1$ for some $(x,y)\in (T \times P)\cup (P
\times T)$, then $(x,y)$ is an arc from the place (transition) $x$
to the transition (place) $y$. The locality mapping $L$ defines sets
of transitions called localities (depending on the number associated
to each transition). The delay mapping $D$ introduces a time delay to
each object created by a transition; the delays indicate how long
the objects cannot be used in other transitions. The initial marking
$M_0$ assigns to each place a number of tokens, and value $0$ to the
global clock $gc$.

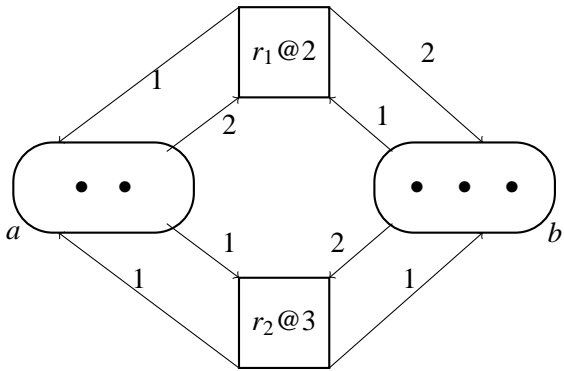
\begin{figure}[ht]
\begin{tabular}{c@{\hspace{4ex}}c}
\begin{tikzpicture}[scale=1.2]
\begin{scope}[xshift=4cm,yshift=1cm]
\draw[thick,rounded corners=15pt] (0.0,0.5) rectangle (2.0,1.5);

\node at (0.0,0.5) {$a$};

\node at (1.0,1.0) {$\bullet\quad \bullet$};
\end{scope}

\begin{scope}[xshift=8cm,yshift=1cm]
\draw[thick,rounded corners=15pt] (0.0,0.5) rectangle (2.0,1.5);

\node at (2.0,0.5) {$b$};

\node at (1.0,1.0) {$\bullet\quad \bullet\quad \bullet$};
\end{scope}

\begin{scope}[xshift=6.5cm,yshift=-0.5cm]
\draw[thick,rounded corners=0pt] (0.0,0.5) rectangle (1.0,1.5);

\node at (0.5,1.0) {$r_2@3$};
\end{scope}

\begin{scope}[xshift=6.5cm,yshift=2.5cm]
\draw[thick,rounded corners=0pt] (0.0,0.5) rectangle (1.0,1.5);

\node at (0.5,1.0) {$r_1@2$};
\end{scope}

\draw[->] (5.7,1.6) -- (6.5,1.0);

\node at (6.2,1.4)[anchor=west] { 1};

\draw[->] (6.5,0.0) -- (4.5,1.5);

\node at (5.2,1.0)[anchor=west] { 1};

\draw[->] (8.2,1.6) -- (7.5,1.0);

\node at (7.4,1.4)[anchor=west] { 2};

\draw[->] (7.5,0.0) -- (9.2,1.5);

\node at (8.2,1.0)[anchor=west] { 1};

\draw[->] (5.7,2.4) -- (6.5,3.0);

\node at (6.2,2.7)[anchor=west] { 2};

\draw[->] (6.5,4.0) -- (4.5,2.5);

\node at (5.4,3.2)[anchor=west] { 1};

\draw[->] (8.2,2.4) -- (7.5,3.0);

\node at (7.9,2.8)[anchor=west] { 1};

\draw[->] (7.5,4.0) -- (9.2,2.5);

\node at (8.4,3.5)[anchor=west] { 2};
\end{tikzpicture}&
\begin{minipage}{7cm}\vspace{-28ex} { Places are drawn as
rounded lines with tokens placed inside. A transition is drawn as a
rectangle containing a label, and the delay it introduces for the
newly created tokens. Transitions are connected to places by
weighted directed arcs.}
\end{minipage}\\
\end{tabular}
\centering\vspace{-2ex}\caption{A Timed Petri Net}
\label{figure:example_petri}
\end{figure}

Markings represent global states of the timed
Petri nets with localities, and they are defined as functions from $P
\cup \{gc\}$ to $\mathds{N}$. A Petri net $\mathcal{N}$ evolves at a
time step $k$ from a marking~$M$ to another marking~$M'$ by a
multiset of transitions $U:T\rightarrow \mathbb{N}$ (e.g., $U(tr)=2$
for $tr\in T$ means that $U$ contains twice the transition $tr$). If
the multiset $U$ of transitions is empty, then the only action is
incrementing the global clock $gc$. Given a multiset of transitions
$U$, we denote by $pre(U)(p)=\sum_{tr\in U} U(tr) \cdot W(p,tr)$ the
multiset of tokens associated to the input arcs ($P\times T$) of all
transitions $tr\in\!U$. In a similar way, by
$post^k_{j}(p)=\sum_{tr\in U;\;D(tr)=j} U(tr) \cdot W(tr,p)$ is
denoted the multiset of tokens associated to the output arcs
($T\times P$) which are added to their corresponding places after
$j$ units of time ($k$ represents the current time). We denote by
$m'=max_{tr\in U}D(tr)$ the maximum delay inferred by the
transitions of $U$. A marking $M$ leads in a max-enabled way to a
marking $M'$ via a multiset $U$ of transitions (denoted by
$M[U\rangle_{max} M'$) if $M'(gc)=M(gc)+1$ and for each place $p\in
P$ the following conditions hold:
\begin{itemize}
\item[$(i)$] $pre(U)(p) \leq M(p)$;
\item[$(ii)$] there is {\bf no} transition $tr\in U$ such that
$pre(\{tr\})(p)+pre(U)(p) \leq M(p)$;
\item[$(iii)$] $M'(p)=M(p)-pre(U)(p)+
\sum^k_{s=max(0,k-m')}post^s_{0}(p)$.
\end{itemize}
According to $(i)$, a marking $M$ has in each place $p$ enough
tokens to enable the execution of the multiset $U$ of transitions.
The maximal parallelism is captured by $(ii)$, saying that an extra
transition cannot be added to $U$. Condition $(iii)$ describes the
effect of the transitions application by adding all the tokens
having $j=0$ created in the last $min(k,m')$ steps which are ready
to be used in Petri nets evolution. Before incrementing the global
clock, all the multisets $post^s_{j}(p)$ are transformed into
$post^{s}_{j-1}(p)$ for $max(0,k-m') \leq s,j \leq k$.

\begin{proposition}
\label{tPNtoPN} For every timed Petri net with localities
$\mathcal{N}=(P,T,W,L,D,M_0)$ there exists a Petri net with
localities $\mathcal{N}'=(P',T',W',L',M'_0)$ that simulates the
evolution of $\mathcal{N}$ (with respect to places of $P$).
Formally, for all $p\in P$ and $k \in \mathds{N}$ we have
$M_k(p)=M'_k(p)$, where $M_k$ and $M'_k$ are markings of $\mathcal{N}$ and
$\mathcal{N}'$ at step~$k$.
\end{proposition}

\begin{proof} In what follows we show how starting from
a timed Petri net with localities $\mathcal{N}=(P,T,W,L,$
$D,M_0)$, we construct an untimed Petri net with localities
$\mathcal{N}'=(P',T',W',L',M'_0)$, where
\begin{itemize}
\item for every $p\in P$ and $tr\in T$ such that $W(p,tr)>0$, we
consider additional places $p,p^{0}_{tr}, \ldots, p^{D(tr)-1}_{tr}$
in $P'$; if $D(tr)=0$ then only $p\in P'$;

\item for every $tr\in T$ and $p\in P$ such that $W(tr,p)>0$, we consider
additional transitions $tr,tr^{0}, \ldots,$ $tr^{D(tr)-1}$ in $T'$; if
$D(tr)=0$ then $tr\in T'$;

\item for every $p\in P$ and $tr\in T$ such that $W(p,tr)>0$,
we consider the weights $W'(p,tr)$ in $\mathcal{N}'$:
\begin{itemize}
\item if $D(tr)=0$ then $W'(p,tr)=W(p,tr)$;
\item if $D(tr)>0$ then $W'(p,tr)=W(p,tr)$, and \\
$W'(tr,p^{D(tr)-1}_{tr})=W'(p^j_{tr},tr^j)=W'(tr^i,p^{i-1}_{tr})=W(tr,p)$
for $0\! \leq\! j<i \!\leq\! D(tr)-1$;
\end{itemize}

\item for every $p\in P$ and $tr\in T$ such that
$W(tr,p)>0$, we consider the following weights $W'(tr,p)$ in
$\mathcal{N}'$:
if $D(tr)=0$ then $W'(tr,p)=W(tr,p)$, else $W'(tr^{0},p)=W(tr,p)$;

\item for every $tr\in T$, we take the same locality label
$l=L(tr)$ for the new transitions $tr,tr^{0}, \ldots, tr^{D(tr)-1}$;

\item if $p \in P$ then
$M'_0(p)=M_0(p)$, and if $p\in P'\backslash P$ then $M'_0(p)=0$ .
\end{itemize}
We show that each step of the timed Petri nets with localities can
be simulated by the corresponding untimed Petri nets with localities;
we prove this by induction on the number of steps (time units)
in timed Petri nets with localities.

Firstly, we consider a marking $M_0$ of the timed Petri net with
localities and a multiset of transitions $U$ such that $M_0[U
\rangle_{max}M_{1}$. The resulting marking $M_{1}$ is given by
$M_1(p)=M_0(p)-pre(U)(p)+post^0_{0}(p)$ for all $p\in P$. Following
the construction above, the initial marking of the untimed Petri net
with localities is $M'_0$, where $M'_0(p)=M_0(p)$ for all $p \in P'$.
$U'$ is the multiset of transitions obtained from $U$ such that
$M'_0[U' \rangle_{max}M'_{1}$. The resulting marking $M'_{1}$ is
given by $M'_1(p')=M'_0(p')-pre(U')(p)+post(U')(p)$ for all $p'\in
P'$, where $post(U')(p)=\sum_{tr\in U'}( U'(tr) \cdot W'(tr,p))$.
This marking contains all the places of $M_1$ and some additional
places from $P'$. Regarding the places $p\in P$, it results that
$post^0_{0}(p)=post(U')(p)$, namely $M'_1(p)=M_1(p)$. Therefore
$M'_1$ equals $M_1$ regarding the number of tokens from the places
of $P$ (we ignore the new places of $P'$ because they do not play
any role at this step).

Secondly, we consider a marking $M_k$ of the timed Petri net with localities
and a multiset $U$ of transitions such that $M_k[U \rangle_{max}M_{k+1}$. The
resulting configuration $M_{k+1}$ is given by $M_{k+1}(p)=M_k(p)-pre(U)(p)+
\sum^k_{s=max(0,k-m')}post^s_{0}(p)$ for all $p\in P$. In the same time, the
multisets of tokens $post^s_{j}(p)$ are renamed by $post^{s}_{j-1}(p)$ for all
$max(k-m',0) \leq s,j \leq k$ and $p\in P$. Following the construction above,
the marking of the untimed Petri net with localities is $M'_k$, where
$M'_k(p)=M_k(p)$ for all $p\in P$, and $M'_k(p^j_{tr})=\sum^{k}_{s=max(0,k-m')}
post^{s}_{j}(p)$ for all additional $p^j_{tr} \in P'\backslash P$, $tr\in T$
and $0\leq j \leq D(tr)-1$. This means that the common places of both nets have
the same number of tokens, while for the additional places appearing only in
$P'$ we add tokens such that for each token from $post^{s}_{j}(p)$ obtained
after firing the transition $tr$, the place $p^{j}_{tr}$ contains a token. The
restriction $max(0,k-m') \leq s\leq k$ (used when creating a token in a new
place $p^{j}_{tr}$ of $P'$) means that a token appears on an output arc of
transition $tr$ in timed Petri nets during the last $min(k,m')$ units of time;
this token has to wait $j$ units of time until it is added to place $p$ of $P$.
The multiset of rules $U'$ is obtained from $U$ such that $M'_k[U'
\rangle_{max}M'_{k+1}$, with $M'_{k+1}(p)=M'_{k}(p)-pre(U')(p)+post(U')(p)$ for
all $p\in P'$. Moreover, in this step some tokens are transferred from places
of $P'$ into places of $P$ by firing the transitions $tr^0$ (the other tokens
from places $p^j_{tr}\in P'$ are transferred into places $p^{j-1}_{tr}\in P'$
by firing the transitions $tr^j$). Thus, the number of tokens obtained in
places $p\in P$ at each step $k$ is equal to $\sum^k_{s=k-m'}post^s_{0}(p)$. It
results that $\sum^k_{s=k-m'}post^s_{0}(p)=post(U')(p)$, namely
$M'_{k+1}(p)=M_{k+1}(p)$ for all $p\in P$. Therefore $M'_{k+1}$ equals
$M_{k+1}$ regarding the number of tokens from the places of $P$ (we ignore the
remaining places $p^j_{tr}\in P'$ because they are used only to simulate the
passage of time).
\end{proof}

\begin{example}
We consider a timed Petri net with localities
$\mathcal{N}=(P,T,W,L,D,M_0)$, where
\begin{itemize}
\item $P=\{(a,1),(a,2),(b,1),(b,2)\}$; \qquad
$T=\{tr^{r_1}_1,tr^{r_2}_2\}$;

\item $D(tr^{r_1}_1)=0$; \qquad $D(tr^{r_2}_2)=2$; \qquad
$L(tr^{r_1}_1)=1$; \qquad $L(tr^{r_2}_2)=2$;

\item $W((a,1),t^{r_2}_2)=W(tr^{r_2}_2,(a,2))=W((b,1),t^{r_1}_1)=W(tr^{r_1}_1,(b,2))=1$

\item $M_0((a,1))=M_0((b,1))=M_0((b,2))=1$; \qquad $M_0((a,2))=2$; \qquad $M_0(gc)=0$.
\end{itemize}

\noindent Graphically the system at time unit $gc=0$ can be
represented as follows:

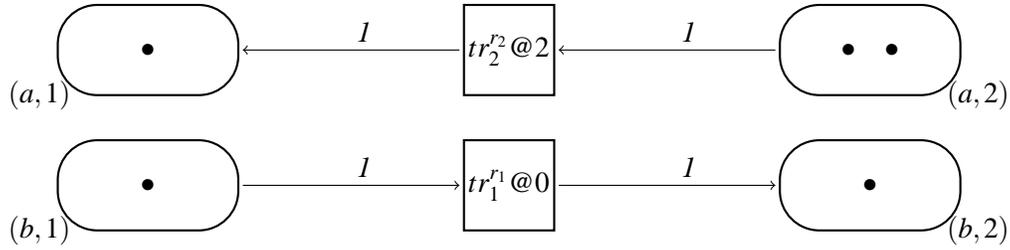
\begin{center}
\begin{tikzpicture}[scale=1.2]
\begin{scope}[xshift=0cm,yshift=0cm]
\begin{scope}[xshift=2cm,yshift=0cm]
\draw[thick,rounded corners=15pt] (0.0,0.5) rectangle (2.0,1.5);

\node at (-0.2,0.5) {$(a,1)$};

\node at (1.0,1.0) {$\bullet$};
\end{scope}

\begin{scope}[xshift=10cm,yshift=0cm]
\draw[thick,rounded corners=15pt] (0.0,0.5) rectangle (2.0,1.5);

\node at (2.2,0.5) {$(a,2)$};

\node at (1.0,1.0) {$\bullet\quad \bullet$};
\end{scope}

\begin{scope}[xshift=6.5cm,yshift=0cm]
\draw[thick,rounded corners=0pt] (0.0,0.5) rectangle (1.0,1.5);

\node at (0.5,1.0) {$tr^{r_2}_2@2$};
\end{scope}

\draw[<-] (4.05,1.0) -- (6.45,1.0);

\node at (5.2,1.2)[anchor=west] { 1};

\draw[<-] (7.55,1.0) -- (9.95,1.0);

\node at (8.8,1.2)[anchor=west] { 1};
\end{scope}

\begin{scope}[xshift=0cm,yshift=-1.5cm]
\begin{scope}[xshift=2cm,yshift=0cm]
\draw[thick,rounded corners=15pt] (0.0,0.5) rectangle (2.0,1.5);

\node at (-0.2,0.5) {$(b,1)$};

\node at (1.0,1.0) {$\bullet$};
\end{scope}

\begin{scope}[xshift=10cm,yshift=0cm]
\draw[thick,rounded corners=15pt] (0.0,0.5) rectangle (2.0,1.5);

\node at (2.2,0.5) {$(b,2)$};

\node at (1.0,1.0) {$\bullet$};
\end{scope}

\begin{scope}[xshift=6.5cm,yshift=0cm]
\draw[thick,rounded corners=0pt] (0.0,0.5) rectangle (1.0,1.5);

\node at (0.5,1.0) {$tr^{r_1}_1@0$};
\end{scope}

\draw[->] (4.05,1.0) -- (6.45,1.0);

\node at (5.2,1.2)[anchor=west] { 1};

\draw[->] (7.55,1.0) -- (9.95,1.0);

\node at (8.8,1.2)[anchor=west] { 1};
\end{scope}
\end{tikzpicture}
\end{center}

\noindent For $gc=1$ and $gc=2$, the timed Petri net with localities can
be represented as follows:

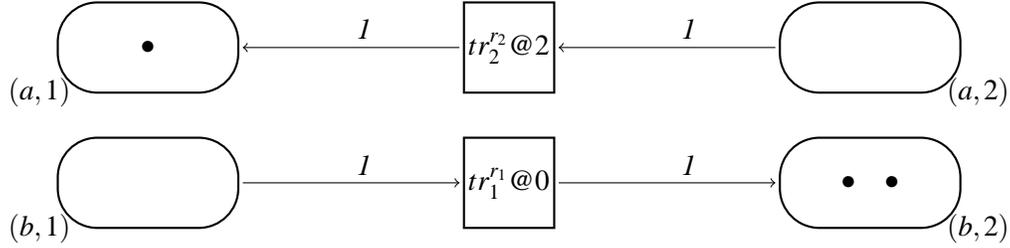
\begin{center}
\begin{tikzpicture}[scale=1.2]
\begin{scope}[xshift=0cm,yshift=0.0cm]
\begin{scope}[xshift=2cm,yshift=0cm]
\draw[thick,rounded corners=15pt] (0.0,0.5) rectangle (2.0,1.5);

\node at (-0.2,0.5) {$(a,1)$};

\node at (1.0,1.0) {$\bullet$};
\end{scope}

\begin{scope}[xshift=10cm,yshift=0cm]
\draw[thick,rounded corners=15pt] (0.0,0.5) rectangle (2.0,1.5);

\node at (2.2,0.5) {$(a,2)$};
\end{scope}

\begin{scope}[xshift=6.5cm,yshift=0cm]
\draw[thick,rounded corners=0pt] (0.0,0.5) rectangle (1.0,1.5);

\node at (0.5,1.0) {$tr^{r_2}_2@2$};
\end{scope}

\draw[<-] (4.05,1.0) -- (6.45,1.0);

\node at (5.2,1.2)[anchor=west] { 1};

\draw[<-] (7.55,1.0) -- (9.95,1.0);

\node at (8.8,1.2)[anchor=west] { 1};
\end{scope}

\begin{scope}[xshift=0cm,yshift=-1.5cm]
\begin{scope}[xshift=2cm,yshift=0cm]
\draw[thick,rounded corners=15pt] (0.0,0.5) rectangle (2.0,1.5);

\node at (-0.2,0.5) {$(b,1)$};

\end{scope}

\begin{scope}[xshift=10cm,yshift=0cm]
\draw[thick,rounded corners=15pt] (0.0,0.5) rectangle (2.0,1.5);

\node at (2.2,0.5) {$(b,2)$};

\node at (1.0,1.0) {$\bullet\quad\bullet$};
\end{scope}

\begin{scope}[xshift=6.5cm,yshift=0cm]
\draw[thick,rounded corners=0pt] (0.0,0.5) rectangle (1.0,1.5);

\node at (0.5,1.0) {$tr^{r_1}_1@0$};
\end{scope}

\draw[->] (4.05,1.0) -- (6.45,1.0);

\node at (5.2,1.2)[anchor=west] { 1};

\draw[->] (7.55,1.0) -- (9.95,1.0);

\node at (8.8,1.2)[anchor=west] { 1};
\end{scope}
\end{tikzpicture}
\end{center}

\noindent while for all $gc\geq3$ we have the following representation

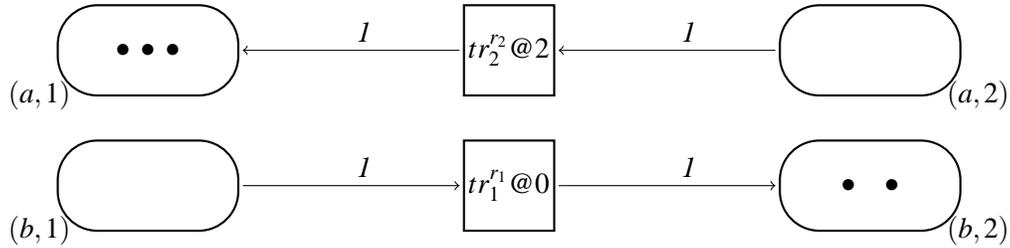
\begin{center}
\begin{tikzpicture}[scale=1.2]
\begin{scope}[xshift=0cm,yshift=0.0cm]
\begin{scope}[xshift=2cm,yshift=0cm]
\draw[thick,rounded corners=15pt] (0.0,0.5) rectangle (2.0,1.5);

\node at (-0.2,0.5) {$(a,1)$};

\node at (1.0,1.0) {$\bullet\;\bullet\;\bullet$};
\end{scope}

\begin{scope}[xshift=10cm,yshift=0cm]
\draw[thick,rounded corners=15pt] (0.0,0.5) rectangle (2.0,1.5);

\node at (2.2,0.5) {$(a,2)$};

\end{scope}

\begin{scope}[xshift=6.5cm,yshift=0cm]
\draw[thick,rounded corners=0pt] (0.0,0.5) rectangle (1.0,1.5);

\node at (0.5,1.0) {$tr^{r_2}_2@2$};
\end{scope}

\draw[<-] (4.05,1.0) -- (6.45,1.0);

\node at (5.2,1.2)[anchor=west] { 1};

\draw[<-] (7.55,1.0) -- (9.95,1.0);

\node at (8.8,1.2)[anchor=west] { 1};
\end{scope}

\begin{scope}[xshift=0cm,yshift=-1.5cm]
\begin{scope}[xshift=2cm,yshift=0cm]
\draw[thick,rounded corners=15pt] (0.0,0.5) rectangle (2.0,1.5);

\node at (-0.2,0.5) {$(b,1)$};

\end{scope}

\begin{scope}[xshift=10cm,yshift=0cm]
\draw[thick,rounded corners=15pt] (0.0,0.5) rectangle (2.0,1.5);

\node at (2.2,0.5) {$(b,2)$};

\node at (1.0,1.0) {$\bullet\quad\bullet$};
\end{scope}

\begin{scope}[xshift=6.5cm,yshift=0cm]
\draw[thick,rounded corners=0pt] (0.0,0.5) rectangle (1.0,1.5);

\node at (0.5,1.0) {$tr^{r_1}_1@0$};
\end{scope}

\draw[->] (4.05,1.0) -- (6.45,1.0);

\node at (5.2,1.2)[anchor=west] { 1};

\draw[->] (7.55,1.0) -- (9.95,1.0);

\node at (8.8,1.2)[anchor=west] { 1};
\end{scope}
\end{tikzpicture}
\end{center}

\noindent We construct an untimed Petri net with localities
$\mathcal{N}'=(P',T',W',L',M'_0)$, where
\begin{itemize}
\item $P=\{(a,1),p,p^0,p^1,(b,1),(b,2)\}$ and
$T=\{tr^{r_1}_1,tr,tr^0,tr^1\}$, where $p=(a,2)$ and
$tr=tr^{r_2}_2$;

\item $L(tr^{r_1}_1)=1$; \qquad $L(tr)=L(tr^0)=L(tr^1)=2$;

\item $W((b,1),t^{r_1}_1)=W(tr^{r_1}_1,(b,2))=1$

\item $W(p,tr)=W(tr,p^1)=W(p^1,tr^1)=W(tr^1,p^0)=W(p^0,tr^0)=W(tr^0,(a,1))=1$

\item $M_0((a,1))=M_0((b,1))=M_0((b,2))=1$; \qquad $M_0(p)=2$; \qquad $M_0(p^0)=M_0(p^1)=0$.
\end{itemize}

\noindent Graphically, the initial system can be represented as follows:

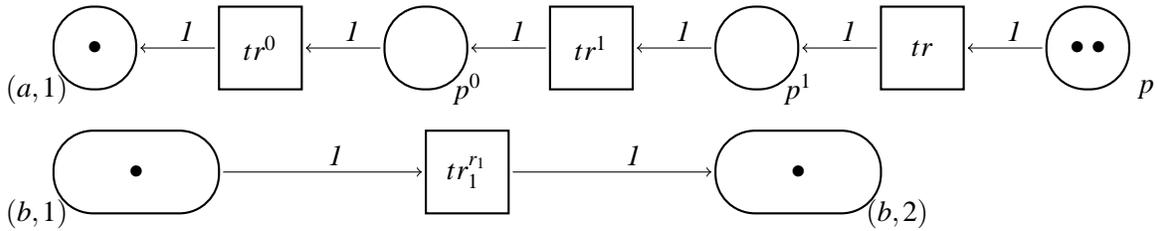
\begin{center}
\begin{tikzpicture}[scale=1.1]
\begin{scope}[xshift=0cm,yshift=0cm]
\begin{scope}[xshift=2cm,yshift=0cm]
\draw[thick,rounded corners=15pt] (0.0,0.5) rectangle (1.0,1.5);

\node at (-0.2,0.5) {$(a,1)$};

\node at (0.5,1.0) {$\bullet$};
\end{scope}

\begin{scope}[xshift=14cm,yshift=0cm]
\draw[thick,rounded corners=15pt] (0.0,0.5) rectangle (1.0,1.5);

\node at (1.2,0.5) {$p$};

\node at (0.5,1.0) {$\bullet\; \bullet$};
\end{scope}

\begin{scope}[xshift=4cm,yshift=0cm]
\draw[thick,rounded corners=0pt] (0.0,0.5) rectangle (1.0,1.5);

\node at (0.5,1.0) {$tr^0$};
\end{scope}

\begin{scope}[xshift=6cm,yshift=0cm]
\draw[thick,rounded corners=15pt] (0.0,0.5) rectangle (1.0,1.5);

\node at (1.0,0.5) {$p^0$};

\end{scope}

\begin{scope}[xshift=8cm,yshift=0cm]
\draw[thick,rounded corners=0pt] (0.0,0.5) rectangle (1.0,1.5);

\node at (0.5,1.0) {$tr^1$};
\end{scope}

\begin{scope}[xshift=10cm,yshift=0cm]
\draw[thick,rounded corners=15pt] (0.0,0.5) rectangle (1.0,1.5);

\node at (1.0,0.5) {$p^1$};
\end{scope}

\begin{scope}[xshift=12cm,yshift=0cm]
\draw[thick,rounded corners=0pt] (0.0,0.5) rectangle (1.0,1.5);

\node at (0.5,1.0) {$tr$};
\end{scope}

\draw[<-] (3.05,1.0) -- (3.95,1.0);

\node at (3.4,1.2)[anchor=west] { 1};

\draw[<-] (5.05,1.0) -- (5.95,1.0);

\node at (5.4,1.2)[anchor=west] { 1};

\draw[<-] (7.05,1.0) -- (7.95,1.0);

\node at (7.4,1.2)[anchor=west] { 1};

\draw[<-] (9.05,1.0) -- (9.95,1.0);

\node at (9.4,1.2)[anchor=west] { 1};

\draw[<-] (11.05,1.0) -- (11.95,1.0);

\node at (11.4,1.2)[anchor=west] { 1};

\draw[<-] (13.05,1.0) -- (13.95,1.0);

\node at (13.4,1.2)[anchor=west] {1};
\end{scope}

\begin{scope}[xshift=0cm,yshift=-1.5cm]
\begin{scope}[xshift=2cm,yshift=0cm]
\draw[thick,rounded corners=15pt] (0.0,0.5) rectangle (2.0,1.5);

\node at (-0.2,0.5) {$(b,1)$};

\node at (1.0,1.0) {$\bullet$};
\end{scope}

\begin{scope}[xshift=10cm,yshift=0cm]
\draw[thick,rounded corners=15pt] (0.0,0.5) rectangle (2.0,1.5);

\node at (2.2,0.5) {$(b,2)$};

\node at (1.0,1.0) {$\bullet$};
\end{scope}

\begin{scope}[xshift=6.5cm,yshift=0cm]
\draw[thick,rounded corners=0pt] (0.0,0.5) rectangle (1.0,1.5);

\node at (0.5,1.0) {$tr^{r_1}_1$};
\end{scope}

\draw[->] (4.05,1.0) -- (6.45,1.0);

\node at (5.2,1.2)[anchor=west] { 1};

\draw[->] (7.55,1.0) -- (9.95,1.0);

\node at (8.8,1.2)[anchor=west] { 1};
\end{scope}
\end{tikzpicture}
\end{center}

\noindent
After one step, we obtain

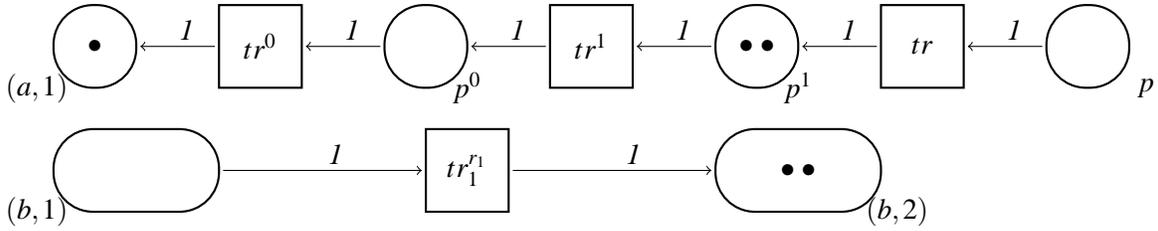
\begin{center}
\begin{tikzpicture}[scale=1.1]
\begin{scope}[xshift=0cm,yshift=0cm]
\begin{scope}[xshift=2cm,yshift=0cm]
\draw[thick,rounded corners=15pt] (0.0,0.5) rectangle (1.0,1.5);

\node at (-0.2,0.5) {$(a,1)$};

\node at (0.5,1.0) {$\bullet$};
\end{scope}

\begin{scope}[xshift=14cm,yshift=0cm]
\draw[thick,rounded corners=15pt] (0.0,0.5) rectangle (1.0,1.5);

\node at (1.2,0.5) {$p$};

\end{scope}

\begin{scope}[xshift=4cm,yshift=0cm]
\draw[thick,rounded corners=0pt] (0.0,0.5) rectangle (1.0,1.5);

\node at (0.5,1.0) {$tr^0$};
\end{scope}

\begin{scope}[xshift=6cm,yshift=0cm]
\draw[thick,rounded corners=15pt] (0.0,0.5) rectangle (1.0,1.5);

\node at (1.0,0.5) {$p^0$};

\end{scope}

\begin{scope}[xshift=8cm,yshift=0cm]
\draw[thick,rounded corners=0pt] (0.0,0.5) rectangle (1.0,1.5);

\node at (0.5,1.0) {$tr^1$};
\end{scope}

\begin{scope}[xshift=10cm,yshift=0cm]
\draw[thick,rounded corners=15pt] (0.0,0.5) rectangle (1.0,1.5);

\node at (1.0,0.5) {$p^1$};

\node at (0.5,1.0) {$\bullet\; \bullet$};
\end{scope}

\begin{scope}[xshift=12cm,yshift=0cm]
\draw[thick,rounded corners=0pt] (0.0,0.5) rectangle (1.0,1.5);

\node at (0.5,1.0) {$tr$};
\end{scope}

\draw[<-] (3.05,1.0) -- (3.95,1.0);

\node at (3.4,1.2)[anchor=west] { 1};

\draw[<-] (5.05,1.0) -- (5.95,1.0);

\node at (5.4,1.2)[anchor=west] { 1};

\draw[<-] (7.05,1.0) -- (7.95,1.0);

\node at (7.4,1.2)[anchor=west] { 1};

\draw[<-] (9.05,1.0) -- (9.95,1.0);

\node at (9.4,1.2)[anchor=west] { 1};

\draw[<-] (11.05,1.0) -- (11.95,1.0);

\node at (11.4,1.2)[anchor=west] { 1};

\draw[<-] (13.05,1.0) -- (13.95,1.0);

\node at (13.4,1.2)[anchor=west] { 1};
\end{scope}

\begin{scope}[xshift=0cm,yshift=-1.5cm]
\begin{scope}[xshift=2cm,yshift=0cm]
\draw[thick,rounded corners=15pt] (0.0,0.5) rectangle (2.0,1.5);

\node at (-0.2,0.5) {$(b,1)$};
\end{scope}

\begin{scope}[xshift=10cm,yshift=0cm]
\draw[thick,rounded corners=15pt] (0.0,0.5) rectangle (2.0,1.5);

\node at (2.2,0.5) {$(b,2)$};

\node at (1.0,1.0) {$\bullet\;\bullet$};
\end{scope}

\begin{scope}[xshift=6.5cm,yshift=0cm]
\draw[thick,rounded corners=0pt] (0.0,0.5) rectangle (1.0,1.5);

\node at (0.5,1.0) {$tr^{r_1}_1$};
\end{scope}

\draw[->] (4.05,1.0) -- (6.45,1.0);

\node at (5.2,1.2)[anchor=west] { 1};

\draw[->] (7.55,1.0) -- (9.95,1.0);

\node at (8.8,1.2)[anchor=west] { 1};
\end{scope}
\end{tikzpicture}
\end{center}

\noindent The system evolves to

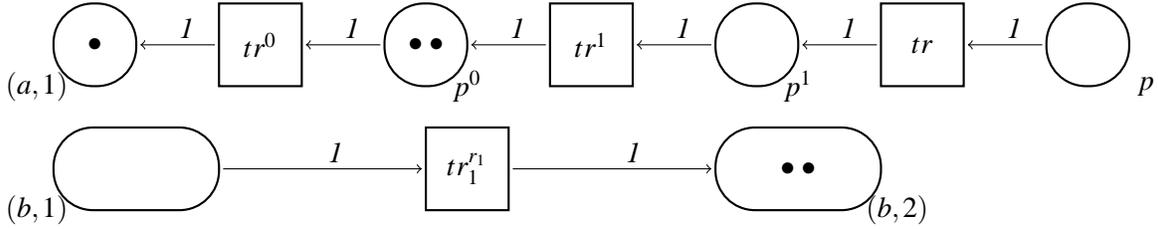
\begin{center}
\begin{tikzpicture}[scale=1.1]
\begin{scope}[xshift=0cm,yshift=0cm]
\begin{scope}[xshift=2cm,yshift=0cm]
\draw[thick,rounded corners=15pt] (0.0,0.5) rectangle (1.0,1.5);

\node at (-0.2,0.5) {$(a,1)$};

\node at (0.5,1.0) {$\bullet$};
\end{scope}

\begin{scope}[xshift=14cm,yshift=0cm]
\draw[thick,rounded corners=15pt] (0.0,0.5) rectangle (1.0,1.5);

\node at (1.2,0.5) {$p$};

\end{scope}

\begin{scope}[xshift=4cm,yshift=0cm]
\draw[thick,rounded corners=0pt] (0.0,0.5) rectangle (1.0,1.5);

\node at (0.5,1.0) {$tr^0$};
\end{scope}

\begin{scope}[xshift=6cm,yshift=0cm]
\draw[thick,rounded corners=15pt] (0.0,0.5) rectangle (1.0,1.5);

\node at (1.0,0.5) {$p^0$};

\node at (0.5,1.0) {$\bullet\; \bullet$};

\end{scope}

\begin{scope}[xshift=8cm,yshift=0cm]
\draw[thick,rounded corners=0pt] (0.0,0.5) rectangle (1.0,1.5);

\node at (0.5,1.0) {$tr^1$};
\end{scope}

\begin{scope}[xshift=10cm,yshift=0cm]
\draw[thick,rounded corners=15pt] (0.0,0.5) rectangle (1.0,1.5);

\node at (1.0,0.5) {$p^1$};

\end{scope}

\begin{scope}[xshift=12cm,yshift=0cm]
\draw[thick,rounded corners=0pt] (0.0,0.5) rectangle (1.0,1.5);

\node at (0.5,1.0) {$tr$};
\end{scope}

\draw[<-] (3.05,1.0) -- (3.95,1.0);

\node at (3.4,1.2)[anchor=west] { 1};

\draw[<-] (5.05,1.0) -- (5.95,1.0);

\node at (5.4,1.2)[anchor=west] { 1};

\draw[<-] (7.05,1.0) -- (7.95,1.0);

\node at (7.4,1.2)[anchor=west] { 1};

\draw[<-] (9.05,1.0) -- (9.95,1.0);

\node at (9.4,1.2)[anchor=west] { 1};

\draw[<-] (11.05,1.0) -- (11.95,1.0);

\node at (11.4,1.2)[anchor=west] { 1};

\draw[<-] (13.05,1.0) -- (13.95,1.0);

\node at (13.4,1.2)[anchor=west] { 1};
\end{scope}

\begin{scope}[xshift=0cm,yshift=-1.5cm]
\begin{scope}[xshift=2cm,yshift=0cm]
\draw[thick,rounded corners=15pt] (0.0,0.5) rectangle (2.0,1.5);

\node at (-0.2,0.5) {$(b,1)$};
\end{scope}

\begin{scope}[xshift=10cm,yshift=0cm]
\draw[thick,rounded corners=15pt] (0.0,0.5) rectangle (2.0,1.5);

\node at (2.2,0.5) {$(b,2)$};

\node at (1.0,1.0) {$\bullet\;\bullet$};
\end{scope}

\begin{scope}[xshift=6.5cm,yshift=0cm]
\draw[thick,rounded corners=0pt] (0.0,0.5) rectangle (1.0,1.5);

\node at (0.5,1.0) {$tr^{r_1}_1$};
\end{scope}

\draw[->] (4.05,1.0) -- (6.45,1.0);

\node at (5.2,1.2)[anchor=west] { 1};

\draw[->] (7.55,1.0) -- (9.95,1.0);

\node at (8.8,1.2)[anchor=west] { 1};
\end{scope}
\end{tikzpicture}
\end{center}

\noindent
The system stops its evolution after reaching the configuration

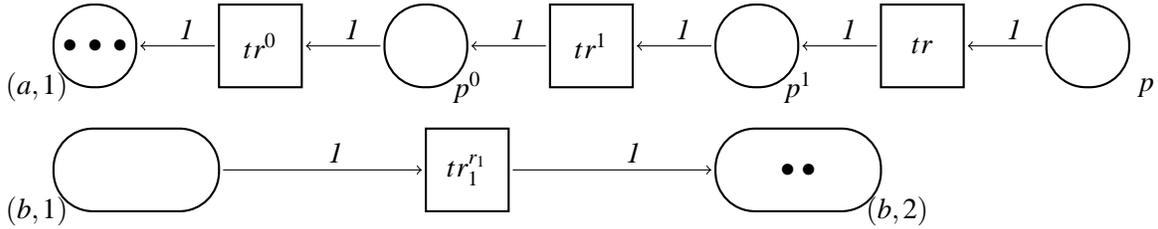
\begin{center}
\begin{tikzpicture}[scale=1.1]
\begin{scope}[xshift=0cm,yshift=0cm]
\begin{scope}[xshift=2cm,yshift=0cm]
\draw[thick,rounded corners=15pt] (0.0,0.5) rectangle (1.0,1.5);

\node at (-0.2,0.5) {$(a,1)$};

\node at (0.5,1.0) {$\bullet\;\bullet\;\bullet$};
\end{scope}

\begin{scope}[xshift=14cm,yshift=0cm]
\draw[thick,rounded corners=15pt] (0.0,0.5) rectangle (1.0,1.5);

\node at (1.2,0.5) {$p$};

\end{scope}

\begin{scope}[xshift=4cm,yshift=0cm]
\draw[thick,rounded corners=0pt] (0.0,0.5) rectangle (1.0,1.5);

\node at (0.5,1.0) {$tr^0$};
\end{scope}

\begin{scope}[xshift=6cm,yshift=0cm]
\draw[thick,rounded corners=15pt] (0.0,0.5) rectangle (1.0,1.5);

\node at (1.0,0.5) {$p^0$};

\end{scope}

\begin{scope}[xshift=8cm,yshift=0cm]
\draw[thick,rounded corners=0pt] (0.0,0.5) rectangle (1.0,1.5);

\node at (0.5,1.0) {$tr^1$};
\end{scope}

\begin{scope}[xshift=10cm,yshift=0cm]
\draw[thick,rounded corners=15pt] (0.0,0.5) rectangle (1.0,1.5);

\node at (1.0,0.5) {$p^1$};

\end{scope}

\begin{scope}[xshift=12cm,yshift=0cm]
\draw[thick,rounded corners=0pt] (0.0,0.5) rectangle (1.0,1.5);

\node at (0.5,1.0) {$tr$};
\end{scope}

\draw[<-] (3.05,1.0) -- (3.95,1.0);

\node at (3.4,1.2)[anchor=west] { 1};

\draw[<-] (5.05,1.0) -- (5.95,1.0);

\node at (5.4,1.2)[anchor=west] { 1};

\draw[<-] (7.05,1.0) -- (7.95,1.0);

\node at (7.4,1.2)[anchor=west] { 1};

\draw[<-] (9.05,1.0) -- (9.95,1.0);

\node at (9.4,1.2)[anchor=west] { 1};

\draw[<-] (11.05,1.0) -- (11.95,1.0);

\node at (11.4,1.2)[anchor=west] { 1};

\draw[<-] (13.05,1.0) -- (13.95,1.0);

\node at (13.4,1.2)[anchor=west] { 1};
\end{scope}

\begin{scope}[xshift=0cm,yshift=-1.5cm]
\begin{scope}[xshift=2cm,yshift=0cm]
\draw[thick,rounded corners=15pt] (0.0,0.5) rectangle (2.0,1.5);

\node at (-0.2,0.5) {$(b,1)$};
\end{scope}

\begin{scope}[xshift=10cm,yshift=0cm]
\draw[thick,rounded corners=15pt] (0.0,0.5) rectangle (2.0,1.5);

\node at (2.2,0.5) {$(b,2)$};

\node at (1.0,1.0) {$\bullet\;\bullet$};
\end{scope}

\begin{scope}[xshift=6.5cm,yshift=0cm]
\draw[thick,rounded corners=0pt] (0.0,0.5) rectangle (1.0,1.5);

\node at (0.5,1.0) {$tr^{r_1}_1$};
\end{scope}

\draw[->] (4.05,1.0) -- (6.45,1.0);

\node at (5.2,1.2)[anchor=west] { 1};

\draw[->] (7.55,1.0) -- (9.95,1.0);

\node at (8.8,1.2)[anchor=west] { 1};
\end{scope}
\end{tikzpicture}
\end{center}

\noindent We notice that indeed, if we refer only to the
markings of the places from $P$ during the evolution of timed and
untimed Petri nets with localities, the markings are the same.

\end{example}

It is easy to prove that the class of timed Petri net with
localities includes the class of Petri net with localities, since we
can assign $0$ to all values of the function~$D$, namely all
transitions fire instantaneously.

\section{Linking Timed Membrane Systems to Timed Petri Nets}
\label{subsection:relationship}

Following the approach given in \cite{Kleijn06} where membrane systems are
translated into Petri nets with localities, we present a translation of timed
membrane systems into timed Petri nets with localities, and then prove an
operational correspondence between them.

\begin{definition}
\label{definition:translation} Let
$\Pi=(V,H,\mu,w_1,\ldots,w_n,R_1,\ldots,R_n,e)$ be a timed membrane
system. Then the corresponding timed Petri net with localities is
$\mathcal{N}_\Pi=(P,T,W,L,D,M_0)$ with its components defined as
follows:

\begin{itemize}
\item $P=V \times \{1,\ldots,n\}$ - to each object $a$ of
membrane $i$ there corresponds a place $p=(a,i)$;

\item $T=\{tr^r_j \mid r\in R_j, 1\leq j \leq n\}$ - to
each rule $r$ of membrane $j$ corresponds a transition $tr^r_j$;

\item for every place $p=(a,i) \in P$ and every transition $tr=tr^r_j \in T$

\begin{math} W(p,tr)= \left\{
{\begin{tabular}{@{\hspace{0ex}}l@{\hspace{2ex}}l@{\hspace{0ex}}}
    $lhs^r_i(a)$ & if $i=j$\\
    $0$ & otherwise\\
\end{tabular}} \right.
\end{math} and
\begin{math} W(tr,p)= \left\{
{\begin{tabular}{@{\hspace{0ex}}l@{\hspace{2ex}}l@{\hspace{0ex}}}
    $rhs^{r,0}_{i,e(r)}(a)$ & if $i=j$\\
    $rhs^{r,0}_{i,e(r)}((a,out))$ & if $(i,j)\in \mu$\\
    $rhs^{r,0}_{i,e(r)}((a,in_j))$ & if $(j,i)\in \mu$\\
    $0$ & otherwise\\
\end{tabular}} \right.
\end{math}

\item for every place $p = (a, i) \in P$, we have $M_0(p) = w_i(a)$;
\item for every transition $tr^r_j \in T$, we have $L(t)=j$;
\item for every $tr^r_j \in T$, we have $D(tr)=e(r)$.
\end{itemize}
\end{definition}

\begin{example}
We consider a timed membrane system $\Pi=(V,\mu,w_1,w_2,R_1,R_2,e)$,
where
\begin{itemize}
\item $V=\{a,b\}$; \qquad $\mu=[[~]_2]_1$; \qquad
$w_1=ab$; \qquad $w_2=a^2b$;
\item $R_1=\{r_1: b \rightarrow (b,in_2)\}$; \qquad
$R_2=\{r_2: a \rightarrow (a,out)\}$; \qquad
 $e(r_1)=0$, $e(r_2)=2$.
\end{itemize}

Graphically, the initial configuration can be depicted as:

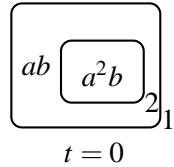
\begin{center}
\begin{tikzpicture}[scale=1.1]
\begin{scope}[xshift=-0.2cm]
\draw[thick,rounded corners=4pt] (1.0,0.5) rectangle (2.8,2.0);

\node at (2.9,0.6) {$1$};

\node at (1.3,1.25) {$ab$};

\begin{scope}[xshift=1.6cm,yshift=-0.2cm]
\draw[thick,rounded corners=4pt] (0.0,1.0) rectangle (1.0,1.75);

\node at (1.1,1.0) {$2$};

\node at (0.5,1.35) {$a^2b$};
\end{scope}

\node at (2,0.2) {$t=0$};
\end{scope}

\end{tikzpicture}
\end{center}

\noindent The corresponding timed Petri net with localities is
$\mathcal{N}=(P,T,W,L,D,M_0)$, where:
\begin{itemize}
\item $P=\{(a,1),(a,2),(b,1),(b,2)\}$; \qquad
$T=\{tr^{r_1}_1,tr^{r_2}_2\}$;

\item $D(tr^{r_1}_1)=0$; \qquad $D(tr^{r_2}_2)=2$; \qquad
$L(tr^{r_1}_1)=1$; \qquad $L(tr^{r_2}_2)=2$;

\item $W((a,1),t^{r_2}_2)=W(tr^{r_2}_2,(a,2))=W((b,1),t^{r_1}_1)=W(tr^{r_1}_1,(b,2))=1$

\item $M_0((a,1))=M_0((b,1))=M_0((b,2))=1$; \qquad $M_0((a,2))=2$; \qquad $M_0(gc)=0$.
\end{itemize}

\noindent Graphically, the system at time unit $gc=0$ can be
represented as

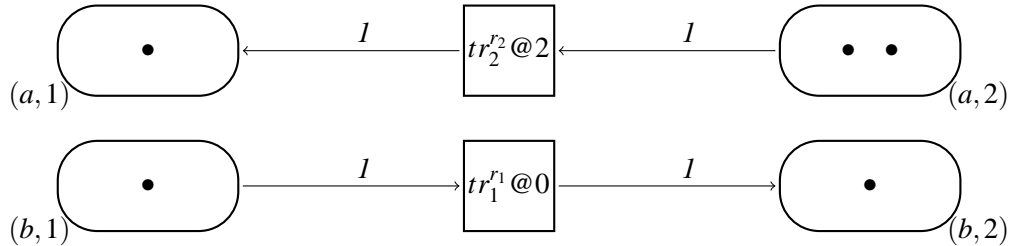
\begin{center}
\begin{tikzpicture}[scale=1.2]
\begin{scope}[xshift=0cm,yshift=0cm]
\begin{scope}[xshift=2cm,yshift=0cm]
\draw[thick,rounded corners=15pt] (0.0,0.5) rectangle (2.0,1.5);

\node at (-0.2,0.5) {$(a,1)$};

\node at (1.0,1.0) {$\bullet$};
\end{scope}

\begin{scope}[xshift=10cm,yshift=0cm]
\draw[thick,rounded corners=15pt] (0.0,0.5) rectangle (2.0,1.5);

\node at (2.2,0.5) {$(a,2)$};

\node at (1.0,1.0) {$\bullet\quad \bullet$};
\end{scope}

\begin{scope}[xshift=6.5cm,yshift=0cm]
\draw[thick,rounded corners=0pt] (0.0,0.5) rectangle (1.0,1.5);

\node at (0.5,1.0) {$tr^{r_2}_2@2$};
\end{scope}

\draw[<-] (4.05,1.0) -- (6.45,1.0);

\node at (5.2,1.2)[anchor=west] { 1};

\draw[<-] (7.55,1.0) -- (9.95,1.0);

\node at (8.8,1.2)[anchor=west] { 1};
\end{scope}

\begin{scope}[xshift=0cm,yshift=-1.5cm]
\begin{scope}[xshift=2cm,yshift=0cm]
\draw[thick,rounded corners=15pt] (0.0,0.5) rectangle (2.0,1.5);

\node at (-0.2,0.5) {$(b,1)$};

\node at (1.0,1.0) {$\bullet$};
\end{scope}

\begin{scope}[xshift=10cm,yshift=0cm]
\draw[thick,rounded corners=15pt] (0.0,0.5) rectangle (2.0,1.5);

\node at (2.2,0.5) {$(b,2)$};

\node at (1.0,1.0) {$\bullet$};
\end{scope}

\begin{scope}[xshift=6.5cm,yshift=0cm]
\draw[thick,rounded corners=0pt] (0.0,0.5) rectangle (1.0,1.5);

\node at (0.5,1.0) {$tr^{r_1}_1@0$};
\end{scope}

\draw[->] (4.05,1.0) -- (6.45,1.0);

\node at (5.2,1.2)[anchor=west] { 1};

\draw[->] (7.55,1.0) -- (9.95,1.0);

\node at (8.8,1.2)[anchor=west] { 1};
\end{scope}
\end{tikzpicture}
\end{center}
\end{example}
According to this translation, $M_C$ denotes the marking
of $\mathcal{N}_\Pi$ corresponding to a configuration $C$ of the
timed membrane system $\Pi$. Moreover, for each multiset $R$ of
applied rules in a timed membrane system, the corresponding multiset
of transitions in timed Petri nets with localities is denoted by
$U_R$. Using these notations, we have the following operational correspondence:

\begin{proposition}\label{proposition:corresp}
$C \stackrel{R}{\Longrightarrow} C'$ if and only if $~M_C[U_R\rangle_{max} M_{C'}$.
\end{proposition}
\begin{proof}[Proof] Let us consider the membrane
configuration $C=(w_1,\ldots,$ $w_n)$. According to Definition
\ref{definition:translation}, we have $M_C(p)=w_i(a)$ for each place
$p=(a,i)$. This is a consequence of the fact that there is a
correspondence between membranes and places, and between the
multiset inside membranes and the marking of the places.
After applying the multiset $R$ of rules in
$C$, we obtain a configuration $C'=(w'_1,\ldots,w'_n)$ where for each
membrane $i$ and each object $a$ we have
$w'_i(a)=w_i(a)-lhs_i(a)+\sum^{k}_{s=max(0,k-m)} rhs^s_{i,0}(a)$. In
the corresponding timed Petri net with localities, starting from the
marking $M_C$ and applying the multiset $U_R$ of transitions, we
obtain a new marking $M'$ where for each place $p$ we get
$M'(p)=M_{C}(p)-pre(U_R)(p)+\sum^k_{s=max(0,k-m')}post^s_{0}(p)$. It
is easy to note that $M_C(p)=w_i(a)$, $pre(U_R)(p)=lhs_i(a)$ and\\
$\sum^k_{s=max(0,k-m')}post^s_{0}(p)=\sum^{k}_{s=max(0,k-m)}
rhs^s_{i,0}(a)$. Therefore, it results that $M'(p)\!=\!w'_i(a)$ and
$M'=M'_C$.
\end{proof}

\section{Conclusion}
\label{section:conclusion}

There exist papers in the field of membrane computing in which the
concept of time is used mainly as timers for objects and
membranes \cite{CompMod09,IJCCC10}, and as execution period for
each rule \cite{Cavaliere05,Cavaliere10}.
The idea of adding time to Petri nets is described in
\cite{Peterson81}: ``addition of timing information might
provide a powerful new feature for Petri nets, but may not be
possible in a manner consistent with the basic philosophy of Petri
nets''. Different ways of incorporating timing information into
Petri nets were proposed by many researchers; specific application
fields represent the inspiration for different proposals of
modelling time. For Petri nets with localities \cite{Kleijn06},
time constrains are added in a way inspired by the coloured Petri nets.

In this paper we prove that adding timing to both membrane
systems and Petri nets with localities does not increase the
expressive power of the corresponding untimed formalisms,
establish a link between these timed formalisms by defining
a relationship between timed formalisms under the assumption of
maximal firing, and prove an operational correspondence between
them. This relationship allows to use the Petri nets
tools to verify certain behavioural properties (reachability,
boundedness, liveness and fairness) of membrane systems.
An attempt to use Petri nets software to simulate timing aspects in membrane
systems is presented in \cite{Profir05}.

As further work, we can mention the use of timed membrane systems to model
some biological systems, while Petri nets tools can be used to analyze and verify
automatically the (timing) behavioural properties of these models.

\medskip

\noindent {\bf Acknowledgements}. The work of Bogdan Aman is
supported by POSDRU/89/1.5/S/49944.


\begin{thebibliography}{1}

\providecommand{\urlalt}[2]{\href{#1}{#2}}
\providecommand{\doi}[1]{doi:\urlalt{http://dx.doi.org/#1}{#1}}

\bibitem{CompMod09}
B. Aman, G. Ciobanu.
\newblock Mutual Mobile Membranes with Timers.
\newblock {\it Electronic Proceedings in Theoretical Computer
Science}, vol.6, 1--15, 2009.
\newblock \doi{10.4204/EPTCS.6.1}

\bibitem{IJCCC10}
B. Aman, G. Ciobanu.
\newblock Adding Lifetime to Objects and Membranes in P Systems.
\newblock {\it Int. Journal of Computers, Communication \&
Control}, vol.5(3), 268--279, 2010.

\bibitem{Andrei07}
O. Andrei, G. Ciobanu, D. Lucanu.
\newblock A Rewriting Logic Framework for Operational Semantics of Membrane
Systems.
\newblock {\it Theoretical Computer Science}, vol.373, 163--181, 2007.
\newblock \doi{10.1016/j.tcs.2006.12.016}

\bibitem{Cavaliere05}
M. Cavaliere, D. Sburlan.
\newblock Time-Independent P Systems.
\newblock {\it Lecture Notes in Computer Science}, vol.3365, 239--258, 2005.
\newblock \doi{10.1007/978-3-540-31837-8\_14}

\bibitem{Cavaliere10}
M. Cavaliere, D. Sburlan.
\newblock Time in Membrane Computing.
\newblock {\it Oxford Handbook of Membrane Computing}, 594--604,
2010.

\bibitem{Ciobanu10}
G. Ciobanu.
\newblock {\it Membrane Computing and Biologically Inspired
Process Calculi}.
\newblock ``A.I.Cuza'' University Press, Ia\c si, 2010.

\bibitem{Ciobanu06}
G.~Ciobanu, Gh.~P\u aun, M.J.~P\'erez-Jim\'enez.
\newblock {\it Applications of Membrane Computing}, Springer, Natural
Computing Series, 2006.
\newblock \doi{10.1007/3-540-29937-8}

\bibitem{CiobanuHandbook}
G. Ciobanu.
\newblock Semantics of P Systems.
\newblock {\it Oxford Handbook of Membrane Computing}, 413--436, 2010.

\bibitem{Zilio04}
S. Dal Zilio, E. Formenti.
\newblock On the Dynamics of PB Systems: a Petri Net View.
\newblock {\it Lecture Notes in Computer Science}, vol.2933,
153--167, 2004.
\newblock \doi{10.1007/978-3-540-24619-0\_11}

\bibitem{Jensen92}
K. Jensen.
\newblock {\it Coloured Petri Nets; Basic Concepts, Analysis Methods
and Practical Use. Vol. 1,2,3}.
\newblock Monographs in Theoretical Computer Science, Springer,
1992,~1994,~1997.

\bibitem{Kleijn10}
J. Kleijn , M. Koutny.
\newblock Petri Nets and Membrane Computing.
\newblock {\it Oxford Handbook of Membrane Computing}, 389--412,
2010.

\bibitem{Kleijn06}
J. Kleijn , M. Koutny, G. Rozenberg.
\newblock Towards a Petri Net Semantics for Membrane Systems.
\newblock {\it Lecture Notes in Computer Science}, vol.3850,
292--309, 2006.
\newblock \doi{10.1007/11603047\_20}

\bibitem{Lodish08}
H. Lodish, A. Berk, P. Matsudaira, C. Kaiser, M. Krieger, M. Scott,
L. Zipursky, J. Darnell.
\newblock {\it Molecular Cell Biology}, 6th Edition, Freeman, 2008.

\bibitem{Merlin74}
P. M. Merlin.
\newblock {\it A Study of the Recoverability of Computing Systems}.
\newblock PhD thesis, Department of Information and Computer Science,
University of California, Irvine, CA, 1974.

\bibitem{Paun02}
Gh.~P\u aun.
\newblock {\it Membrane Computing. An Introduction}.
\newblock Springer, 2002.
\newblock \doi{10.1007/BF03037282}

\bibitem{Peterson81}
J.L. Peterson.
\newblock {\it Petri Net Theory and the Modeling of Systems}.
\newblock Prentice-Hall, 1981.

\bibitem{Pezze99}
M. Pezz\'e and M. Young.
\newblock Time Petri Nets: A Primer Introduction.
\newblock {\it Multi-Workshop on Formal Methods in Performance Evaluation and
Applications}, 1999.

\bibitem{Profir05}
A. Profir, E. Gu\c tuleac, E. Boian.
\newblock Encoding Continuous-Time P Systems with Descriptive Times Petri
Nets.
\newblock {\it TAPS'05, IEEE Computer Press}, 91--94, 2005.

\bibitem{Ramchandani74}
C. Ramchandani.
\newblock {\it Analysis of Asynchronous Concurrent Systems by Timed Petri Nets}.
\newblock PhD thesis, Massachusetts Institute of Technology, Cambridge, MA,
1974.

\bibitem{Sifakis80}
J. Sifakis.
\newblock Performance Evaluation of Systems using Nets.
\newblock {\it Lecture Notes in Computer Science}, vol. 84, 307--319, 1980.
\newblock \doi{10.1007/3-540-10001-6\_30}

\bibitem{Qi04}
Z. Qi, J. You, H. Mao.
\newblock P Systems and Petri Nets.
\newblock {\it Lecture Notes in Computer Science}, vol.2933,
286--303, 2004.
\newblock \doi{10.1007/978-3-540-24619-0\_21}

\end{thebibliography}
\end{document}